\documentclass[a4paper,final]{amsart}

\usepackage{amsmath,amssymb}
\usepackage{amsthm}
\usepackage{url}
\usepackage{tikz, pgfplots}
\usepackage{multirow}
\usepackage[mathlines]{lineno}
\usepackage[left=30mm,right=30mm,top=35mm,bottom=35mm]{geometry}

\newtheorem{theorem}{Theorem}
\newtheorem{lemma}{Lemma}
\newtheorem{remark}{Remark}
\newtheorem{proposition}{Proposition}

\newcommand{\wt}{\widetilde}
\renewcommand{\epsilon}{\varepsilon}
\renewcommand{\phi}{\varphi}

\begin{document}

\title{Factorizations of the Fibonacci Infinite Word}

\author{Gabriele Fici}
\address{Dipartimento di Matematica e Informatica\\
                Universit\`a di Palermo\\
                Palermo\\
                Italy}
\email{Gabriele.Fici@unipa.it}
\thanks{Published on \textit{Journal of Integer Sequences}, Vol. 18 (2015), Article 15.9.3.}

\begin{abstract}
The aim of this note is to survey the factorizations of the Fibonacci infinite word that make use of the Fibonacci words and other related words, and to show that all these factorizations can be easily derived in sequence starting from elementary properties of the Fibonacci numbers. 
\end{abstract}

\maketitle


\section{Preliminaries}

The well-known sequence of Fibonacci numbers (sequence A000045 in
the {\it On-Line Encyclopedia of Integer Sequences}) is defined by $F_{1}=1$,  $F_{2}=1$ and for every $n>2$, $F_{n}=F_{n-1}+F_{n-2}$.
The first few values of the sequence $F_{n}$ are reported in Table \ref{tab:Fibonumbers} for reference.

\begin{table}[H]
\begin{center}
\scalebox{0.8}{
\begin{tabular}{cccccccccccccccccccccc}
$n$ & 1 & 2& 3& 4& 5& 6& 7& 8& 9& 10& 11 & 12 & 13 & 14 & 15 & 16 & 17 & 18 & 19 & 20 \\
\hline 
$F_{n}$ & 1 & 1 & 2& 3& 5& 8& 13& 21& 34& 55& 89& 144 & 233 & 377 &610 & 987 & 1597 & 2584 & 4181 & 6765
\end{tabular}
}
\caption{\label{tab:Fibonumbers}The first few values of the sequence of Fibonacci numbers.}
\end{center}
\end{table}

A basic property of Fibonacci numbers (that can be easily proved by induction) is that  $1$ plus the sum of the first $n$ Fibonacci numbers is equal to the $(n+2)$-th Fibonacci number:

\begin{equation}\label{eq:Fib}
 1+\sum_{i=1}^nF_i=F_{n+2}.
\end{equation}

We recall here a famous result, usually attributed to Zeckendorf
\cite{Zeck}, but published earlier by Lekkerkerker \cite{Lek} and which,
in fact, is a special case of an older and more general result due to
Ostrowski \cite{Ostr}.  It permits us to use Fibonacci numbers as a
basis for representing integers:

\begin{theorem}\label{theor:zeck}
 Every positive integer can be expressed uniquely as the sum of one or more distinct non-consecutive Fibonacci numbers $F_n$, $n>1$.
\end{theorem}

For example, $17=13+3+1=F_7+F_4+F_2$,
and there is no other way to write $17$ as the sum of non-consecutive Fibonacci numbers (assuming the convention that $F_1$ is not used in the representation). 
Thus, one can represent natural numbers as strings of $0$-$1$ bits, where the $i$-th bit (from the right) encodes the presence/absence of the $(i+1)$-th Fibonacci number in the representation given by Theorem \ref{theor:zeck}. So for example the number $17$ is represented by $100101$. We call this representation of natural numbers the \emph{Zeckendorf representation}.

The first few natural numbers and their Zeckendorf representations are displayed in Table \ref{tab:zeck}, where we padded to the left with $0$s in order to have strings of the same length. Note that with $6$ bits one can represent the first $21$ natural numbers. In fact, for every $n>0$, there are exactly $F_n$ integers whose leftmost $1$ in the Zeckendorf representation is in position $n$ (starting from the right). From $(\ref{eq:Fib})$, we derive that one needs $n$ bits to represent the first $F_{n+2}$ natural numbers. 

The strings of length $n$ forming the Zeckendorf representations of the first $F_{n+2}$ natural numbers are precisely all the $0$-$1$ strings of length $n$ not containing two consecutive $1$s. These strings are in lexicographic order if the natural numbers are in increasing order from $0$ to $F_{n+2}-1$.

\begin{table}[ht]
\centering 
\begin{raggedright}
\begin{tabular}{l *{1}{@{\hspace{2mm}}c@{\hspace{6mm}}} l *{1}{@{\hspace{2mm}}c@{\hspace{6mm}}} l *{1}{@{\hspace{2mm}}c@{\hspace{6mm}}} l *{1}{@{\hspace{2mm}}c}}
Zeck.\  & decimal & Zeck.\   & decimal & Zeck.\   & decimal \\
\hline \rule[-2pt]{0pt}{3pt}\\
$000000$ & 0 & $010000$ & 8 & $100100$ & 16  \\
$000001$ & 1 & $010001$ & 9 & $100101$ & 17  \\
$000010$ & 2 & $010010$ & 10 & $101000$ & 18  \\
$000100$ & 3 & $010100$ & 11 & $101001$ & 19  \\
$000101$ & 4 & $010101$ & 12 & $101010$ & 20 \\
$001000$ & 5 & $100000$ & 13 &  \\
$001001$ & 6 & $100001$ & 14 &  \\
$001010$ & 7 & $100010$ & 15 & \\
\hline \vspace{4mm}
\end{tabular}
\end{raggedright}
\caption{\label{tab:zeck}The Zeckendorf representations of the first few natural numbers coded with 6 bits.}
\end{table}

Let us define $f(n)$, for every $n\geq 0$, as the rightmost digit of the Zeckendorf representation of $n$. For every $n>1$ we define the \emph{$n$-th Fibonacci word} as the string \[f_n=f(0)f(1)\cdots f(F_n-1)\] of length $|f_n|=F_n$. By convention, we set $f_1=1.$ 
The first few Fibonacci words are shown in Table \ref{tab:Fibowords01}.

\begin{table}[ht]
\centering  
\begin{equation*}
\begin{split}
  f_{1}  &= 1 \\
  f_{2}  &= 0 \\
  f_{3}  &= 01 \\
  f_{4}  &= 010 \\
  f_{5}  &= 01001 \\
  f_{6}  &= 01001010 \\
  f_{7}  &= 0100101001001 \\
  f_{8}  &= 010010100100101001010 \\
  f_{9}  &= 0100101001001010010100100101001001 \\
 \end{split}
 \end{equation*}
\caption{\label{tab:Fibowords01}The first few Fibonacci words.}
\end{table}

We also define the \emph{Fibonacci infinite word} $f$ as the limit of $f_n$ as $n$ goes to infinity. That is, $f$ is the infinite word whose $n$-th letter is the ``parity'' of the Zeckendorf representation of $n$: \[f=f(0)f(1)f(2)f(3)\cdots=0100101001001010010\cdots\]

In the Zeckendorf representation of an integer, when the $n$-th digit from the right is a $1$, the $(n-1)$-th digit from the right is a $0$. Hence, the rightmost $n-2$ digits of the Zeckendorf representations of the natural numbers from $F_{n+1}$ to $F_{n+2}-1$ are the same rightmost $n-2$ digits of the Zeckendorf representations  of the first $F_{n}$ natural numbers. For example, the $2$ rightmost digits of the Zeckendorf representations of $5$, $6$ and $7$ are, respectively, $00$, $01$, $10$, as well as the $2$ rightmost digits of the Zeckendorf representations of the $0$, $1$ and $2$. We deduce that for every $n>2$, one has 
\begin{equation}\label{eq:rec}
 f_n=f_{n-1}f_{n-2}.
\end{equation}
For more details on Fibonacci words the reader can see, for instance, \cite{Ber}.

Recall that a \emph{factorization} of an infinite word $w$ is a sequence $(x_n)_{n\geq 1}$ of finite words such that $w$ can be expressed as the concatenation of the elements of the sequence, i.e., $w=\prod_{n\geq 1} x_{n}$.

In general, exhibiting a factorization $(x_n)_{n\geq 1}$ of an infinite word $w$ can be useful to better understand the combinatorics of $w$, provided the sequence $(x_n)_{n\geq 1}$ has non-trivial combinatorial properties---for example, all the words in the sequence are palindromes, squares, or prefixes of $w$. 

Another point of view consists in defining a factorization by some general rule that can be applied to any infinite word. The sequence $(x_n)_{n\geq 1}$ is therefore determined by the particular instance of the infinite word $w$ (as is the case, for example, in the Lempel-Ziv factorization or in the Lyndon factorization, that we will see below). In this case, the word $w$ can have particular properties that make it a limit example for that particular factorization. 

In next sections, we will show a number of factorizations of the Fibonacci infinite word that make use of the Fibonacci finite words and other related words. These factorizations have been introduced over the time in different papers, and we think it can be useful to collect them all together for reference. We also add some (at least to the best of our knowledge) novel factorizations. Moreover, we present these factorizations in an order that allows us to provide a short and elementary proof for each of them, despite the original proofs being sometimes more involved or more technical.

\section{Fibonacci words and co-Fibonacci words}

The first factorization  of the Fibonacci infinite word we exhibit is the following.

\begin{proposition}
The Fibonacci infinite word can be obtained by concatenating $0$ and the Fibonacci words:
\begin{eqnarray}\label{fibo}
 f &=& 0\prod_{n\geq 1} f_{n}\\
 &=& 0 \cdot 1 \cdot 0 \cdot 01 \cdot 010 \cdot 01001 \cdot 01001010 \cdots \notag
\end{eqnarray}
\end{proposition}

\begin{proof}
 Since for every $i\geq 1$, $|f_{i}|=F_{i}$, it is sufficient to prove that, for every $n\geq 1$, $f_{n}$ occurs in $f$ starting at position $1+\sum_{i=1}^{n-1}F_{i}=F_{n+1}$. From $(\ref{eq:rec})$, we have $f_{n+2}=f_{n+1}f_n$, so that $f_{n}$ has an occurrence in $f$ starting at position $|f_{n+1}|=F_{n+1}$.
\end{proof}

Let us consider the sequence $p_{n}$ of the palindromic prefixes of $f$, also called \emph{central words}. The first few values of the sequence $p_n$ are displayed in Table \ref{tab:central}, where $\epsilon$ denotes the empty word, i.e., the word of length $0$.

\begin{table}[ht]
\centering  
\begin{equation*}
\begin{split}
 p_{3} &= \epsilon\\
 p_{4} &= 0\\
 p_{5} &= 010\\
 p_{6} &= 010010\\
 p_{7} &= 01001010010\\
 p_{8} &= 0100101001001010010\\
 p_{9} &= 01001010010010100101001001010010\\
 \end{split}
 \end{equation*}
\caption{\label{tab:central}The first few central words.}
\end{table}

As it is well-known, for every $n\geq 3$, $p_{n}$ is obtained from $f_{n}$ by removing the last two letters. More precisely, we have for every $n\geq 1$, 
\begin{equation}\label{central}
f_{2n+1}=p_{2n+1}01, \hspace{6mm} f_{2n+2}=p_{2n+2}10.
\end{equation}

The fundamental property of the central words is the following:

\begin{lemma}\label{lem:central}
For every $n\geq 2$ one has
\[
 p_{2n+1}=p_{2n-1}01p_{2n}=p_{2n}10p_{2n+1}, \hspace{8mm}  p_{2n+2}=p_{2n}10p_{2n+1}=p_{2n+1}01p_{2n}.
\]
\end{lemma}

\begin{proof}
 Follows immediately from $(\ref{eq:rec})$ and $(\ref{central})$.
\end{proof}

\begin{remark}
It is easy to see from $(\ref{eq:rec})$ that for every $n\geq 4$, one has $f_{n}=f_{n-2}f_{n-3}f_{n-2}$.
We have therefore from $(\ref{fibo})$:
\begin{eqnarray}\label{multisingular}
 f &=& 01001 \prod_{n\geq 2}f_{n}f_{n-1}f_{n}\\
 &=& 01001 \cdot (0\cdot 1\cdot 0)(01\cdot 0\cdot 01)(010\cdot 01\cdot 010)(01001\cdot 010 \cdot 01001)  \cdots \notag
\end{eqnarray}
Analogously, since $1=f_3$, we can write
\begin{eqnarray}\label{multisingular2}
f &=& 0100 \prod_{n\geq 2}f_{n-1}f_{n}f_{n-1}\\
 &=& 0100 \cdot (1\cdot 0\cdot 1)(0\cdot 01\cdot 0)(01\cdot 010\cdot 01)(010\cdot 01001 \cdot 010) \cdots \notag
\end{eqnarray}
\end{remark}

We now introduce a class of words that we call the {\it co-Fibonacci words}. Although this class has appeared previously in the literature \cite{BerCoFib}, to the best of our knowledge no one has yet given a name to them.

The co-Fibonacci words $f'_{n}$ are defined by complementing the last two letters in the Fibonacci words $f_{n}$, that is, $f'_{n}=p_{n}yx$, where $x$ and $y$ are the letters such that $f_{n}=p_{n}xy$. Equivalently, co-Fibonacci words can be defined by $f'_{n}=f_{n-2}f_{n-1}$ for every $n\geq 3$. The first few co-Fibonacci words are displayed in Table \ref{tab:Cofibowords}.

\begin{table}[ht]
\centering  
\begin{equation*}
\begin{split}
 f'_{3} &= 10\\
 f'_{4} &= 001\\
 f'_{5} &= 01010\\
 f'_{6} &= 01001001\\
 f'_{7} &= 0100101001010\\
 f'_{8} &= 010010100100101001010\\
 \end{split}
 \end{equation*}
\caption{\label{tab:Cofibowords}The first few co-Fibonacci words.}
\end{table}

The following lemma is a direct consequence of Lemma \ref{lem:central}.

\begin{lemma}\label{lem:cofib}
 For every $n\geq 2$ one has
\[
 f'_{2n+1}=f_{2n}f'_{2n-1}, \hspace{8mm}  f_{2n+2}=f_{2n}f'_{2n+1}.
\]
\end{lemma}

\begin{proposition}
The Fibonacci word can be obtained by concatenating $0$ and the odd co-Fibonacci words:
\begin{eqnarray}
 f &=& 0 \prod_{n\geq 1}f'_{2n+1}\\
 &=& 0 \cdot 10 \cdot 01010 \cdot 0100101001010 \cdots \notag
\end{eqnarray}
\end{proposition}

\begin{proof}
Follows directly from (\ref{fibo}) replacing $f_{2n-1}f_{2n}$ with $f'_{2n+1}$.
\end{proof}

Analogously, we have the following:

\begin{proposition}
The Fibonacci word can be obtained by concatenating $01$ and the even co-Fibonacci words:
\begin{eqnarray}
  f &=& 01 \prod_{n\geq 1}f'_{2n+2}\\
 &=& 01 \cdot 001 \cdot 01001001 \cdot 010010100100101001010 \cdots \notag
\end{eqnarray}
\end{proposition}

\begin{proof}
Follows directly from (\ref{fibo}) replacing $f_{2n}f_{2n+1}$ with $f'_{2n+2}$.
\end{proof}

\section{Singular words}

Let us define the \emph{left rotation} of a non-empty word $w=w_1w_2\cdots w_n$, $w_i$ letters, as the word $w^{\lambda}=w_nw_1\cdots w_{n-1}$. Analogously, the \emph{right rotation} of $w$ is defined as the word $w^{\rho}=w_2\cdots w_{n}w_1$.

The \emph{singular words} $\hat{f}_{n}$ are defined by complementing the first letter in the left rotations of the Fibonacci words $f_{n}$. The first few singular words are displayed in Table \ref{tab:Singular}. Note that for every $n\geq 1$, one has $\hat{f}_{2n+1}=0p_{2n+1}0$ and $\hat{f}_{2n+2}=1p_{2n+2}1$ .

\begin{table}[ht]
\centering  
\begin{equation*}
\begin{split}
 \hat{f}_{1} &= 0\\
 \hat{f}_{2} &= 1\\
 \hat{f}_{3} &= 00\\
 \hat{f}_{4} &= 101\\
 \hat{f}_{5} &= 00100\\
 \hat{f}_{6} &= 10100101\\
 \end{split}
 \end{equation*}
\caption{\label{tab:Singular}The first few singular words.}
\end{table}

The singular words are palindromic factors of $f$ but do not appear as prefixes of $f$ (by the way, $f$ also contains  other palindromic factors besides the central words $p_{n}$ and the singular words $\hat{f}_{n}$, e.g., $1001$ or $01010$, see \cite{DelDel06a} for more details). Their name comes from the fact that  among the $F_{n}+1$ factors of $f$ of length $F_{n}$, there are $F_{n}$ of them that can be obtained one from each other by iteratively applying left (or equivalently right) rotation and one, the singular word, whose left (or equivalently right) rotation is not a factor of $f$.

Wen and Wen \cite{WenWen94} proved that the Fibonacci infinite word can be obtained by concatenating the singular words:

\begin{proposition}
The Fibonacci infinite word is the concatenation of the singular words:
\begin{eqnarray}\label{singular}
 f &=& \prod_{n\geq 1}\hat{f}_{n}\\
 &=& 0 \cdot 1 \cdot 00 \cdot 101 \cdot 00100 \cdot 10100101 \cdots \notag
\end{eqnarray}
\end{proposition}

\begin{proof}
Indeed, (\ref{singular}) follows directly from (\ref{fibo}) and the definition of singular words, observing that the Fibonacci words end by letter  $0$ and $1$ alternatingly. 
\end{proof}

The factorization (\ref{singular}) is in fact the Lempel-Ziv factorization of $f$. The Lempel-Ziv factorization is a factorization widely used in computer science for compressing strings \cite{LZ77}. The Lempel-Ziv factorization  of a word $w$ is $w=w_{1}w_{2}\cdots$ where $w_1$ is the first letter of $w$  and for every $i\geq 2$, $w_i$ is the shortest prefix of $w_iw_{i+1}\cdots$ that occurs only once in the word $w_1w_2\cdots w_i$.
 Roughly speaking, at each step one searches for the shortest factor that did not appear before.

\begin{remark}
It is easy to see, using for example (\ref{eq:rec}), (\ref{central}) and the definition of singular words,  that for every $n\geq 4$, $\hat{f}_{n}=\hat{f}_{n-2}\hat{f}_{n-3}\hat{f}_{n-2}$.
Therefore, from (\ref{singular}), we have
\begin{eqnarray}\label{multisingular3}
 f &=& 0100 \prod_{n\geq 2}\hat{f}_{n}\hat{f}_{n-1}\hat{f}_{n}\\
 &=& 0100 \cdot (1\cdot 0\cdot 1)(00\cdot 1\cdot 00)(101\cdot 00\cdot 101)(00100\cdot 101 \cdot 00100)  \cdots \notag
\end{eqnarray}
Since $0=\hat{f}_1$, we hence obtain
\begin{eqnarray}\label{multisingular4}
f &=& 010 \prod_{n\geq 2}\hat{f}_{n-1}\hat{f}_{n}\hat{f}_{n-1}\\
 &=& 010 \cdot (0\cdot 1\cdot 0)(1\cdot 00\cdot 1)(00\cdot 101\cdot 00)(101\cdot 00100 \cdot 101) \cdots \notag
\end{eqnarray}

The factorization (\ref{multisingular4}) is a sort of dual with Lucas numbers of the factorization in singular words (\ref{singular}). Indeed, the sequence of factor lengths in (\ref{singular}) is the sequence of Fibonacci numbers, while if in (\ref{multisingular4}) one decomposes the first term as  $01\cdot 0$, then the sequence of factor lengths is the sequence of Lucas numbers (sequence A000032 in
the {\it On-Line Encyclopedia of Integer Sequences}): 2,1,3,4,7,11, etc. 
\end{remark}

\section{Christoffel words}\label{sec:Lyndon}

The \emph{lower Christoffel words} are defined by $c_{n}=0p_{n}1$, for every $n\geq 3$. The lower Christoffel words are the Lyndon factors of $f$, i.e., they are lexicographically smaller than any of their proper suffixes (with respect to the order induced by $0<1$). 

\begin{table}[ht]
\centering  
\begin{equation*}
\begin{split}
 c_{3} &= 01\\
 c_{4} &= 001\\
 c_{5} &= 00101\\
 c_{6} &= 00100101\\
 c_{7} &= 0010010100101\\
 c_{8} &= 001001010010010100101\\
 \end{split}
 \end{equation*}
\caption{\label{tab:Lower}The first few lower Christoffel words.}
\end{table}

If in the Euclidean plane one interprets each $0$ by a horizontal unitary step and each $1$ with a vertical unitary step, the lower Christoffel word $c_n$ is the best grid approximation from below of the segment joining the point $(0,0)$ to the point $(F_{n-1},F_{n-2})$ (see Figure \ref{fig:Chris}). 

\begin{figure}
\begin{center}
\begin{minipage}{7.2cm}
\includegraphics[height=45mm]{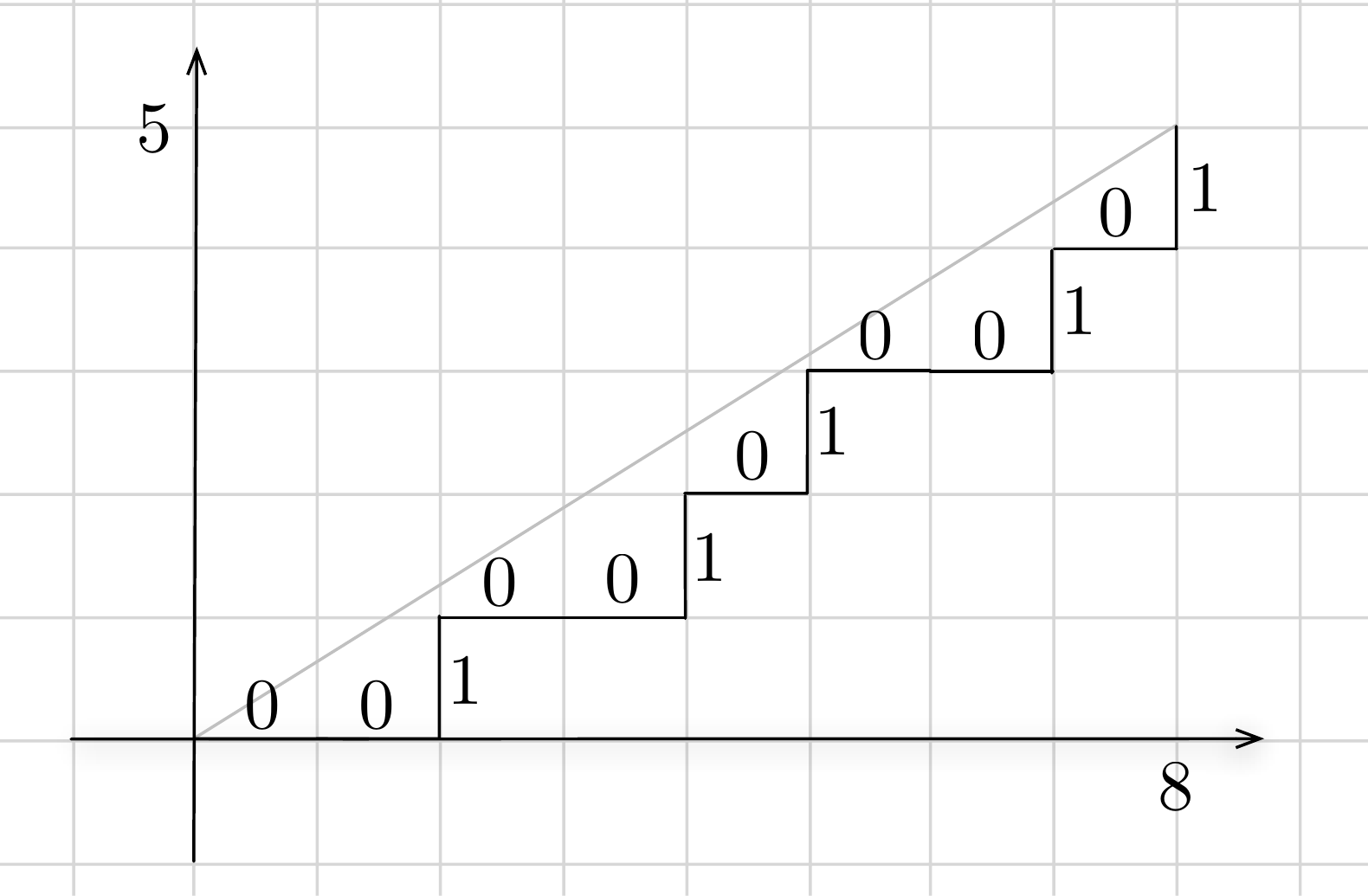}
\end{minipage}
\begin{minipage}{7.2cm}
\includegraphics[height=45mm]{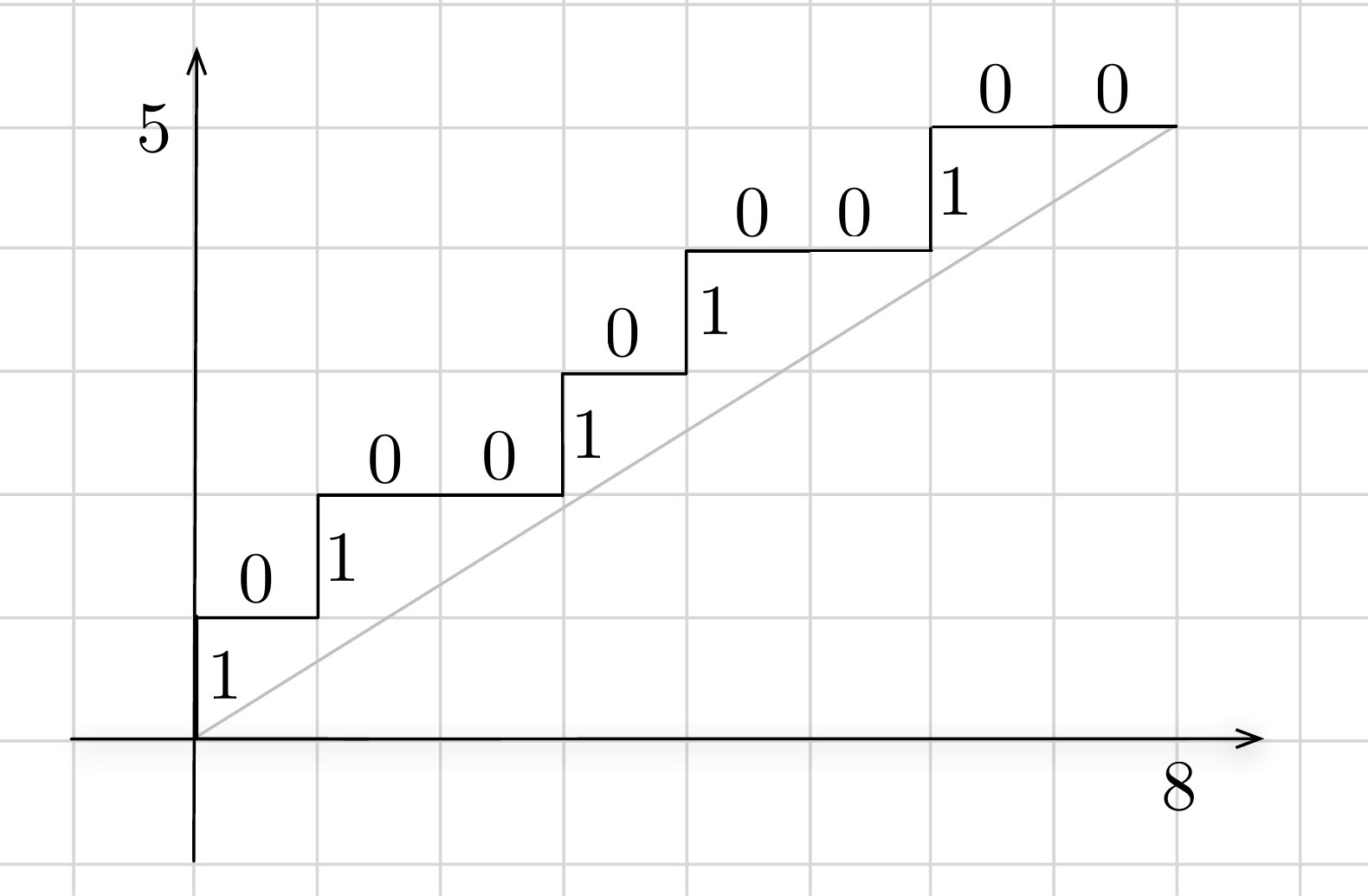}
\end{minipage}
\end{center}
\caption{The lower Christoffel word $c_7=0010010100101$ (left) and the upper Christoffel word $\wt{c_{7}}=1010010100100$ (right) are the best grid approximations, respectively from above and from below, of the Euclidean segment joining the points $(0,0)$ and $(8,5)=(F_6,F_5)$.\label{fig:Chris}}
\end{figure}

Analogously, the \emph{upper Christoffel words} are defined by $\wt{c_{n}}=1p_{n}0$, for every $n\geq 3$. Therefore, the upper Christoffel words are the reversals of the lower Christoffel words (we use the notation $\wt{w}$ for the reversal, a.k.a.~mirror image, of the word $w$).
The upper Christoffel words are the anti-Lyndon factors of $f$, i.e., they are lexicographically greater than any of their proper suffixes (with respect to the order induced by $0<1$). Moreover, $\wt{c_{n}}$ is the best grid approximation from above of the segment joining the point $(0,0)$ to the point $(F_{n-1},F_{n-2})$.

\begin{remark}
 For every $n\geq 3$, $c_n=f^{\lambda}_n$ if $n$ is even, $c_n=\wt{f^{\lambda}_n}$ if $n$ is odd. Therefore, $\wt{c_n}=\wt{f^{\lambda}_n}$ if $n$ is even, $\wt{c_n}=f^{\lambda}_n$ if $n$ is odd.
\end{remark}

\begin{table}[ht]
\centering  
\begin{equation*}
\begin{split}
 \wt{c_{3}} &= 10\\
 \wt{c_{4}} &= 100\\
 \wt{c_{5}} &= 10100\\
 \wt{c_{6}} &= 10100100\\
 \wt{c_{7}} &= 1010010100100\\
 \wt{c_{8}} &= 101001010010010100100\\
 \end{split}
 \end{equation*}
  \vspace{4mm}
\caption{\label{tab:Upper}The first few upper Christoffel words.}
\end{table}

\begin{lemma}\label{lem:Berstel}
For every $n\geq 2$ one has
\[
 c_{2n+1}=c_{2n}c_{2n-1}, \hspace{8mm}  c_{2n+2}=c_{2n}c_{2n+1},
\]
 and therefore
\[
 \wt{c_{2n+1}}=\wt{c_{2n-1}}\wt{c_{2n}}, \hspace{8mm} \wt{c_{2n+2}}=\wt{c_{2n+1}}\wt{c_{2n}}.
\]
\end{lemma}

\begin{proof}
 The first part follows from Lemma \ref{lem:cofib} by applying the right rotation to each side of the equalities.
 The second part follows from the first by applying the reversal.
\end{proof}

The following result states that every Christoffel word is the product of two singular words.

\begin{lemma}\label{lem:Chris}
For every $n\geq 1$ one has
\begin{itemize}
 \item $c_{2n+1}=\hat{f}_{2n-1}\hat{f}_{2n}$ 
 \item $c_{2n+2}=\hat{f}_{2n+1}\hat{f}_{2n}$ 
\end{itemize}
and therefore
\begin{itemize}
\item $\wt{c_{2n+1}}=\hat{f}_{2n}\hat{f}_{2n-1}$
\item $\wt{c_{2n+2}}=\hat{f}_{2n}\hat{f}_{2n+1}$
\end{itemize}
\end{lemma}

\begin{proof}
Follows directly from Lemma \ref{lem:central} and the definitions of Christoffel and singular words.
\end{proof}

Melan\c{c}on \cite{Me99} proved that the Fibonacci word is the concatenation of the odd lower Christoffel words: 

\begin{proposition}
The Fibonacci word is the concatenation of the odd lower Christoffel words:
\begin{eqnarray}\label{Chris}
 f &=&  \prod_{n\geq 1}c_{2n+1}\\
&=& 01 \cdot 00101 \cdot 0010010100101 \cdots\notag
\end{eqnarray}
\end{proposition}

\begin{proof}
Follows directly from (\ref{singular}) and Lemma \ref{lem:Chris}.
\end{proof}

Actually, Melan\c{c}on proved that (\ref{Chris}) is precisely the Lyndon factorization of $f$. Recall that the Lyndon factorization of a word $w$ is $w=\ell_{1}\ell_{2}\cdots$, where each $\ell_{i}$ is a Lyndon word and is lexicographically greater than or equal to $\ell_{i+1}$. The uniqueness of such a factorization for finite words is a well-known theorem of Chen, Fox and Lyndon \cite{CFL58}. Siromoney et al.~extended this factorization to infinite words \cite{SiMaDaSu94}.

Symmetrically, we have the following:

\begin{proposition}
The Fibonacci word is the concatenation of $0$ and the even upper Christoffel words:
\begin{eqnarray}\label{antiChris}
 f &=& 0 \prod_{n\geq 2}\wt{c_{2n}}\\
 &=& 0 \cdot 100 \cdot 10100100  \cdot 101001010010010100100  \cdots \notag
\end{eqnarray}
\end{proposition}

\begin{proof}
 Follows directly from (\ref{singular}) and Lemma \ref{lem:Chris}.
\end{proof}

In fact, it is easy to see that (\ref{antiChris}) is the Lyndon factorization of $f$ if one takes the order induced by  $1<0$.

We now present two other factorizations based on Christoffel words. To the best of our knowledge, these factorizations did not appear before in literature.

\begin{proposition}
The Fibonacci word is the concatenation of $010$ and the lower Christoffel words where each odd lower Christoffel word is squared:
\begin{eqnarray}
 f &=&  010 \prod_{n\geq 1}c_{2n+1}^{2}c_{2n+2}\\
 &=& 010 \cdot (01 \cdot 01 \cdot 001) (00101 \cdot 00101 \cdot 00100101) \cdots \notag
\end{eqnarray}
\end{proposition}

\begin{proof}
Follows directly from (\ref{multisingular2}) and Lemma \ref{lem:Chris}. Indeed, by Lemma \ref{lem:Chris}, we have $$c_{2n+1}c_{2n+1}c_{2n+2}=\hat{f}_{2n-1}\hat{f}_{2n}\cdot \hat{f}_{2n-1}\hat{f}_{2n}\cdot \hat{f}_{2n+1}\hat{f}_{2n}=(\hat{f}_{2n-1}\hat{f}_{2n} \hat{f}_{2n-1})(\hat{f}_{2n}\hat{f}_{2n+1}\hat{f}_{2n}).$$
\end{proof}

Analogously, we have the following:

\begin{proposition}
The Fibonacci word is the concatenation of $0100$ and the upper Christoffel words  where each even upper Christoffel word is squared:
\begin{eqnarray}
 f &=&  0100 \prod_{n\geq 1}\wt{c_{2n+1}}\wt{c_{2n+2}}^{2}\\
 &=& 0100 \cdot (10 \cdot 100 \cdot 100) (10100 \cdot 10100100 \cdot 10100100)\cdots \notag
\end{eqnarray}
\end{proposition}

\begin{proof}
Follows directly from (\ref{multisingular}) and Lemma \ref{lem:Chris}. Indeed, by Lemma \ref{lem:Chris}, we have $$\wt{c_{2n+1}}\wt{c_{2n+2}}\wt{c_{2n+2}}=\hat{f}_{2n}\hat{f}_{2n-1}\cdot \hat{f}_{2n}\hat{f}_{2n+1}\cdot \hat{f}_{2n}\hat{f}_{2n+1}=(\hat{f}_{2n}\hat{f}_{2n-1} \hat{f}_{2n})(\hat{f}_{2n+1}\hat{f}_{2n}\hat{f}_{2n+1}).$$
\end{proof}

\section{Reversals of Fibonacci words}

One of the most known factorizations of the Fibonacci infinite word, and perhaps the most surprising, is the following.

\begin{proposition}
The Fibonacci word can be obtained also by concatenating the reversals of the Fibonacci words:
\begin{eqnarray}\label{rev}
 f &=& \prod_{n\geq 2}\wt{f_{n}}\\
 &=& 0 \cdot 10 \cdot 010 \cdot 10010 \cdot 01010010 \cdots \notag
\end{eqnarray}
\end{proposition}

\begin{proof}
It follows from the definitions that taking the right rotation of $\wt{f_{n}}$ and complementing the last letter produces the $n$-th singular word $\hat{f}_{n}$. Therefore, (\ref{rev}) follows directly from (\ref{singular}) observing that the reversals of the Fibonacci words start with $0$ and $1$ alternatingly.
\end{proof}

The factorization (\ref{rev}) is basically the Crochemore factorization of $f$---the only difference is that the Crochemore factorization starts with $0$, $1$, $0$ and then coincides with the one above (see \cite{BeSa06}). Recall that the Crochemore factorization of $w$ is $w=c_{1}c_{2}\cdots$ where $c_{1}$ is the first letter of $w$ and for every $i>1$, $c_{i}$ is either a fresh letter or the longest prefix of $c_{i}c_{i+1}\cdots$ occurring twice in $f_{1}f_{2}\cdots f_{i}$. For example, the Crochemore factorization of the word $w = 0101001$ is $0\cdot 1\cdot 010\cdot 01$, since $010$ occurs twice in $01010$.  

\medskip

In 1995, de Luca \cite{del95} considered the following factorization:

\begin{proposition}
The Fibonacci word can be obtained by concatenating the reversals of the even Fibonacci words.
\begin{eqnarray}\label{del}
 f &=& \prod_{n\geq 2}\wt{f_{2n}}\\
 &=& 010 \cdot 01010010 \cdot 010100101001001010010 \cdots \notag
\end{eqnarray}
\end{proposition}

\begin{proof}
Applying the reversal to $(\ref{eq:rec})$, we have that $\wt{f_{n}}=\wt{f_{n-2}}\wt{f_{n-1}}$, for every $n>2$. So (\ref{del}) follows directly from (\ref{rev}) by replacing $\wt{f_{2n-2}}\wt{f_{2n-1}}$ with $\wt{f_{2n}}$.
\end{proof}

In \cite{del95} de Luca proved that the factorization (\ref{del}) has the following 
minimal property with respect to the lexicographical order: any non-trivial permutation of a finite number of the factors will produce an infinite word that is lexicographically greater than $f$.

\medskip

Concatenating the reversals of the odd Fibonacci words instead of even ones still produces the Fibonacci word, if one prepends a $0$:

\begin{proposition}
The Fibonacci word can be obtained by concatenating $0$ and the reversals of the odd Fibonacci words:
\begin{eqnarray}
 f &=& 0\prod_{n\geq 2}\wt{f_{2n+1}}\\
 &=& 0 \cdot 10010 \cdot 1001001010010 \cdots \notag
\end{eqnarray}
\end{proposition}

\begin{proof}
Follows directly from (\ref{rev}) by replacing $\wt{f_{2n-1}}\wt{f_{2n}}$  with $\wt{f_{2n+1}}$.
\end{proof}

Recently  \cite{DelFi13}, studying the so-called \emph{open} and \emph{closed} words, the following factorization has been proved:

\begin{proposition}
The Fibonacci word can be obtained by concatenating $01$ and the squares of the reversals of the Fibonacci words:
\begin{eqnarray}
 f &=& 01\prod_{n\geq 2}(\wt{f_{n}})^{2}\\
 &=& 01\cdot (0\cdot 0)(10 \cdot 10)(010 \cdot 010)(10010\cdot 10010)\cdots \notag
\end{eqnarray}
\end{proposition}

\begin{proof}
Recalling that for every $n\geq 3$, one has $\wt{f_{n}}=\wt{f_{n-2}}\wt{f_{n-1}}$, we have, from  (\ref{rev}), that 
$f= \wt{f_{2}}\wt{f_{3}}\wt{f_{4}}\cdots=0\cdot \wt{f_{1}}\wt{f_{2}}\wt{f_{2}}\wt{f_{3}}\cdots=01\cdot (\wt{f_{2}}\wt{f_{2}})(\wt{f_{3}}\wt{f_{3}})\cdots=01\prod_{n\geq  2}(\wt{f_{n}})^{2}$.
\end{proof}

\section{Generalization to standard Sturmian words}

The Fibonacci word is the most prominent example of a \emph{standard Sturmian word}. Let $\alpha$ be an irrational number such that $0<\alpha<1$, and let $\left[0;d_{1}+1,d_{2},d_3,\ldots\right]$ be the continued fraction expansion of $\alpha$.
The sequence of words defined by $s_{1}=1$, $s_{2}=0$ and $s_{n}=s_{n-1}^{d_{n-2}}s_{n-2}$ for $n\geq 3$, converges to the infinite word $w_{\alpha}$, called the standard Sturmian word of slope $\alpha$. The sequence of words $s_{n}$ is called the standard sequence of $w_{\alpha}$. The Fibonacci word is the standard Sturmian word of slope $\alpha=(3-\sqrt{5})/2$ and its standard sequence is the sequence of Fibonacci finite words, since one has $d_i=1$ for every $i\ge 1$.

Most of the factorizations we described in this note can be generalized to any standard Sturmian word. However, the proofs become more technical and less easy to present.

\end{document}